\documentclass[letterpaper, 10 pt, journal, twoside]{ieeetran} 
\usepackage{color}
\usepackage{psfrag}
\usepackage{comment}
\usepackage{balance}
\usepackage{hyperref}
\hypersetup{
  colorlinks   = true, 
  urlcolor     = blue, 
  linkcolor    = blue, 
  citecolor   = blue 
}
\usepackage{mathrsfs}
\usepackage{graphicx}
\usepackage{amsmath,amssymb,amsfonts}

\usepackage{amsthm}
\usepackage{mathalfa}
\usepackage{tikz}
\usetikzlibrary{calc}
\usepackage{pgfplots}
\pgfplotsset{compat=newest}
\usetikzlibrary{plotmarks}
\usetikzlibrary{arrows.meta}
\usepgfplotslibrary{patchplots}
\usepackage{grffile}
\usepackage{amsmath}
\newcommand*{\QEDB}{\hfill\ensuremath{\square}}%
\newtheorem{rem}{Remark}
\newenvironment{remark}
  {\pushQED{\qed}\rem}
  {\popQED\endrem}
\newtheorem{lemma}{Lemma}
\newtheorem{proposition}{Proposition}

\newtheorem{theorem}{Theorem}

\newtheorem{property}{Property}

\newtheorem{definition}{Definition}

\usepackage{graphicx}
\usepackage{amsmath,amssymb,amsfonts,yhmath}
\usepackage{caption}
\usepackage{subcaption}
\usepackage{tikz}
\usetikzlibrary{circuits.ee.IEC}
\usetikzlibrary{arrows}
\DeclareMathOperator*{\argmin}{arg\,min}
\renewcommand{\vec}{\operatorname{vec}}

\newcommand{\id}{\mathbf{id}}
\newcommand{\R}{\mathbb{R}}

\newcommand*{\tr}{^{\mkern-1.5mu\mathsf{T}}}

\newcommand{\dom}{\operatorname{dom}}

\newcommand{\Id}{\mathbf{I}}
\newcommand{\minimize}{\operatorname{minimize}}
\newcommand{\Tr}{\operatorname{trace}}

\newcommand{\SO}{\mathrm{SO}(3)}
\newcommand{\Sone}{\mathrm{SO}(2)}
\newcommand{\A}{\mathcal{A}}

\newcommand{\rk}{{\bf rk}}

\DeclareMathOperator{\rge}{rge}

\newcommand{\Rcal}{\mathcal{R}}
\newcommand{\s}{\vec}

\newcommand{\source}{{THIS IS A PREPRINT VERSION. IF YOU FOUND THIS READING USEFUL FOR YOUR RESEARCH PLEASE CITE THE PUBLISHED VERSION DOI: \href{https://doi.org/10.1109/LCSYS.2021.3090034}{https://doi.org/10.1109/LCSYS.2021.3090034}}}

\usepackage{fancyhdr}
\usepackage{mathtools} 
\makeatletter
\def\ps@IEEEtitlepagestyle{}
\title{\LARGE \bf On Angular Speed Estimation of Rigid Bodies\\ (Extended Version)}
\author{Francesco~Ferrante, \IEEEmembership{Member, IEEE} and Gildas~Besan\c{c}on, \IEEEmembership{Member, IEEE}
\thanks{Francesco Ferrante and  Gildas Besan\c{c}on are with Univ. Grenoble Alpes, CNRS, Grenoble INP, GIPSA-lab, 38000 Grenoble, France. Email: \{francesco.ferrante, gildas.besancon\}@gipsa-lab.fr.}
\thanks{\textcolor{blue}{This file contain corrections to some minor oversights in the published version and adds a clarification in the proof of Theorem 1. Corrections are colored in blue and there are footnotes explaining them. \textbf{Corrections are minor and the results in the printed version of the paper are not affected}. Last update \today.}}
\thanks{\textcolor{red}{In the printed version of our paper, some occurrences of the symbol $\vec$ got lost due to typesetting issues. An erratum has been issued; see \href{https://doi.org/10.1109/LCSYS.2021.3129614}{https://doi.org/10.1109/LCSYS.2021.3129614}}.}
}
\begin{document}
\maketitle
\begin{abstract}            
The problem of estimating the angular speed of a solid body from attitude measurements is addressed. To solve this problem, we propose an observer whose dynamics are not constrained to evolve on any specific manifold. This drastically simplifies the analysis of the proposed observer. Using Lyapunov tools, sufficient conditions for global asymptotic stability of a set wherein the estimation error is equal to zero are established. In addition, the proposed methodology is adapted to deal with angular speed estimation for systems evolving on the unit circle. The approach is illustrated through several numerical simulations.  
\end{abstract}
\section{Introduction}
\subsection{Background}
\IEEEPARstart{H}{andling} rigid bodies, as  in aerial or underwater vehicle applications, has been a source of control challenges for several decades now.  In particular, the problem of estimating the state of a system evolving on the special group $\SO$ has attracted the interest of researchers; see, e.g., \cite{haydar2017h,mahony2008nonlinear} and the references therein. It is worth noting that although the problem of attitude filtering has received a lot of attention, only a few contributions addressing 
angular speed estimation can be found. In \cite{salcudean1991globally} for instance, an observer based on quaternion description and discontinuous correction is proposed, and in \cite{wu2015angular}, an observer is designed directly on $\SO$, providing almost global asymptotic stability. In \cite{berkane2016global}, another design on $\SO$ can be found, providing global convergence. The strategy therein relies on a hybrid mechanism reminiscent of the results in \cite{mayhew2013synergistic}. Recently, angular speed estimators, yet for system spinning around a fixed axis have been presented in \cite{bernard2020estimation,brentari2018class}.
\subsection{Contributions and outline of the paper}
The main contribution of this paper pertains to a novel observer for angular speed estimation of rigid bodies from attitude measurements on $\SO$ expressed in an inertial frame. The key feature of the solution we propose consists of avoiding the common paradigm requiring to confine angular speed observers on the manifold $\SO$. This enables to considerably simplify the structure and the convergence analysis of the observer, which does not require the use of sophisticated geometrical tools. The leading idea of our work is that estimating the angular speed does not require to estimate the attitude. This makes our approach unique when compared to recent results \cite{berkane2016global,wu2015angular}. The price to pay to obtain such a simplification is that our solution does not enable to obtain an estimate of the attitude \cite{wu2015angular}.

The proposed methodology is specialized to the case of a body rotating around a fixed axis. This leads to an angular speed observer that can be exploited in the presence of ``wrapped'' angular measurements, which is a case of practical relevance. As an additional contribution, we show that in this setting, the observer can be coupled to a specific projection map and be also used for angular position filtering. 

The remainder of the paper is organized as follows. Section~\ref{sec:ProbStatSO3} presents the problem we solve and the outline of the proposed observer. Section~\ref{sec:Stability} provides sufficient conditions to ensure global asymptotic stability of the estimation error dynamics and characterizes robustness to small measurement noise. Section~\ref{sec:SO2} specializes the results presented in the previous sections to the specific, yet practically relevant, case of a body spinning around a fixed axis, and presents a thorough discussion about angular position estimation. Section~\ref{sec:Ex} illustrates the proposed methodology in two numerical examples. 
\subsection{Notation and Preliminaries}
The symbol $\R_{\geq 0}$ ($\R_{>0}$) denotes the set of nonnegative (positive) reals, $\R^n$ is the Euclidean space of dimension $n$, $\R^{n\times m}$ is the vector space of $n\times m$ real matrices represented by the canonical basis, $\mathbb{S}_+^n$ is the set of $n\times n$ symmetric positive definite matrices, and $\id$ is the identity function. Given two topological spaces $X$ and $Y$, $\mathcal{C}^0(X;Y)$ is the set of continuous functions from $X$ to $Y$. For a matrix $M\in\mathbb{R}^{n\times m}$ (vector $x\in\R^n$), $M\tr$ ($x\tr$) denotes the transpose of $M$ (of $x$). For a symmetric matrix $M$, $M\succ 0$ denotes positive definiteness of $M$. Given a metric space $\mathcal{M}$, the symbol $\mathbb{B}$ denotes the closed unit ball in $\mathcal{M}$. For Euclidean vectors, we use the equivalent notation $(x, y)=[x\tr\,\,y\tr]\tr$. Let $\mathbb{U}$ be a normed linear vector space, $\vert x\vert$ stands for the norm of $x$. Let $x\in\mathbb{U}$ and $\mathcal{A}\subset \mathbb{U}$ be nonempty, the distance of $x$ to $\mathcal{A}$ is defined as $d(x,\mathcal{A})\coloneqq\displaystyle\inf_{y\in {\mathcal{A}}} \vert x-y \vert$.
Let $A\in\R^{n\times n}$ and $B\in\R^{n\times n}$, $\langle A, B\rangle_{\mathcal{F}}\coloneqq\Tr(A\tr B)$ denotes the Frobenius inner product, and $\vert A\vert_{\mathcal{F}}=\sqrt{\langle A, A\rangle_{\mathcal{F}}}$ the corresponding induced norm. The symbol $\mathfrak{so}(3)$ stands for the set of $3\times 3$ skew matrices and for any positive integer $n$, $\mathrm{SO}(n)\coloneqq\{R\in\R^{n\times n}\colon R\tr R=\Id, \det(R)=1\}$, where $\Id$ denotes the identity matrix. Given $x\in\R^3$, we define $[x]_\times\coloneqq\left[\begin{smallmatrix}
0 & -x_3& x_2\\ x_3 & 0 & -x_1\\ -x_2 & x_1 & 0
\end{smallmatrix}\right].$ 
Notice that $x\mapsto [x]_\times$ is invertible on $\mathfrak{so}(3)$. 
In particular, we denote the inverse of $x\mapsto [x]_\times$ by $\vec\colon\mathfrak{so}(3)\rightarrow\R^3$, defined as $\vec([x]_\times)=x$. Let $\mathcal{U}$ be a finite dimensional  real inner product space with dimension $n$, $V\colon\dom V\subset\mathcal{U}\rightarrow\R$, $x\in\mathcal{U}$, and $\{e_i\}_{i=1}^n$ be an orthonormal basis of $\mathcal{U}$, we use the notation
$\nabla V(x)\coloneqq \sum_{i=1}^n (DV(x)e_i)e_i$, where $DV(x)\colon\mathcal{U}\rightarrow\R$ stands for the differential of $V$ at $x$. For $f\colon X\rightarrow Y$, $\rge f$ stands for the image of $f$. The symbol $V^{-1}(c)$ stands for the $c$-level set of the function $V\colon\dom V\rightarrow \R$.

Next we state two useful properties. The first can be established by inspection.
\begin{property}
Let $\Theta\in\R^{3\times 3}$ and $\omega\in\R^{3}$. Then, 
$$\Tr\left([\omega]_\times\Theta\right)=\omega\tr \s(\Theta\tr-\Theta)$$
\label{prop:TraceProp}
\hfill$\diamond$
\end{property}
\begin{property}
Let $\Theta\in\R^{3\times 3}$ and $\omega\in\R^{3}$. Assume that $\Theta$ is nonsingular.  Then,  
$[\omega]_\times\Theta=0\iff \omega=0$.
\label{prop:Ker}
\end{property}
\begin{proof} 
The implication $\omega=0\implies [\omega]_\times\Theta=0 $
is trivial. To conclude the proof, observe that for all $\omega\in\R^3\setminus\{0\}$, $\rk[\omega]_\times=2$.
Therefore, for all $\omega\in\R^3\setminus\{0\}$, $\dim\ker[\omega]_\times=1$. This implies that for any set of linear independent vectors $\mathcal{V}\coloneqq\{v_1, v_2, v_3\}$, one has $\mathcal{V}\not\subset\ker[\omega]_\times$. This concludes the proof. 
\end{proof}
\subsection{Preliminaries on Constrained Differential Inclusions}
\label{sec:prel}
In this paper we consider dynamical systems of the form:
\begin{equation}
\dot{x}\in F(x),\qquad x\in C,
\label{eq:diffInc}
\end{equation}
where $x\in\mathcal{U}$ is the system state, $\mathcal{U}$ is a normed finite dimensional linear vector space, $C\subset\mathcal{U}$, and the set valued map $F\colon \dom F\rightrightarrows\mathcal{U}$. In particular, we say that \eqref{eq:diffInc} satisfies the so-called \emph{basic conditions} if: $C$ is closed in $\mathcal{U}$, $F$ is outer semicontinuous and locally bounded on $C$, $C\subset\dom F$, and for all $x\in C$, $F(x)$ is convex. A function $\phi\colon \dom\phi\rightarrow\mathcal{U}$ is a solution to \eqref{eq:diffInc} if $\dom\phi$ is an interval of $\R_{\geq 0}$ including zero, $\phi(0)\in C$, $\phi$ is locally absolutely continuous, and for almost all $t\in\dom\phi$, $\phi(t)\in C$ and $\dot{\phi}(t)\in F(\phi(t))$. A solution to \eqref{eq:diffInc} is said to be maximal if its domain cannot be extended and complete if its domain is unbounded. The following stability notion for compact sets is considered.
\begin{definition}
Let $\A\subset \mathcal{U}$ be compact. We say that $\A$ is:
\begin{itemize}
    \item[($i$)] stable for \eqref{eq:diffInc} if for any $\varepsilon>0$, there exists $\delta>0$ such that any solution $\phi$ to \eqref{eq:diffInc}, with $d(\phi(0), \A)\leq \delta$ satisfies for all $t\in\dom\phi$, $d(\phi(t), \A)\leq \varepsilon$;
    \item[($ii$)] globally pre-attractive for \eqref{eq:diffInc} if for any $\mu>0$ every solution $\phi$ to \eqref{eq:diffInc} with $d(\phi(0),\A)\leq \mu$ is bounded and, if $\phi$ is complete, $
    \displaystyle\lim_{t\rightarrow\infty}d(\phi(t), \A)=0
    $;
      \item[($iii$)] globally pre-asymptotically stable (GpAS) for \eqref{eq:diffInc} if it is both stable and globally pre-attractive.
\end{itemize}
\end{definition}
The following definition is used in the paper.
\begin{definition}
Given a nonempty set $\A$ and a function $V\,:\,\dom V\rightarrow\R$, with $\dom V\supset\A$, we say that $V$ is positive definite with respect to $\A$ if for all $x\in(\dom V\setminus\A)$, $V(x)>0$ and $V(\A)=\{0\}$.
\end{definition}
See \cite{goebel2012hybrid} for more details on these assumptions and definitions for the more general case of hybrid dynamical systems. 
\section{Problem Statement and Solution Outline}
\label{sec:ProbStatSO3}
We consider the attitude dynamics of a rigid body in $\R^3$ actuated via a torque $u\in\R^3$, that is: 
\begin{equation}
\label{eq:plant}
\begin{aligned}
&\dot{\Rcal}=[\Rcal J_0^{-1}\Rcal\tr q]_\times\Rcal\\
&\dot{q}=u\\
&\omega=\Rcal J_0^{-1}\Rcal\tr q,
\end{aligned}
\end{equation}
where $\Rcal\in \SO$, $q\in\R^3$, and $\omega\in\R^3$ represent the body attitude, angular momentum, and angular speed (all) expressed in an inertial reference frame, respectively; see, e.g., \cite{wu2015angular}. The matrix $J_0\in\mathbb{S}_+^{3\times 3}$ represents the body inertia matrix expressed in a
body-fixed frame. We assume that the input torque is selected such that the angular momentum $q$ belongs to some compact set $\mathcal{K}_q\subset\R^3$, i.e., the angular speed is bounded. 
Assuming that a measurement of $\Rcal$ is accessible, our goal is to design an observer providing an estimate $\hat{\omega}$ of the angular speed $\omega$. To solve this problem, we propose the following observer with state\footnote
{\textcolor{blue}{$\Rcal$ and $\Rcal\tr$ are swapped in \eqref{eq:obs}}.} $(\widehat{\Rcal}, \hat{q})\in\R^{3\times 3}\times \R^3$:
\begin{equation}
\label{eq:obs}
\begin{aligned}
&\dot{\widehat{\Rcal}}=[\Rcal J_0^{-1}\Rcal\tr\hat{q}]_\times\Rcal+\Gamma(\widetilde{\Rcal})\\
&\dot{\hat{q}}=u+K\Rcal J_0^{-1}\Rcal\tr\s(\widetilde{\Rcal}\Rcal\tr-\Rcal\widetilde{\Rcal}\tr)
\\
&\hat{\omega}=\textcolor{blue}{\Rcal J_0^{-1}\Rcal\tr}\hat{q},
\end{aligned}
\end{equation}
where 
$\widetilde{\Rcal}\coloneqq \Rcal-\widehat{\Rcal}$, while $K\in\R^{3\times 3}$ and $\Gamma\in\mathcal{C}^0(\R^{3\times 3};\R^{3\times 3})$ are to be designed.

The state $\widehat{\Rcal}$ represents a sort of estimate of the attitude $\Rcal$, while $\hat{q}$ represents an estimate of the angular momentum. In this sense, the observer we propose is a full-order observer. On the other hand, it is worthwhile to remark that the state $\widehat{\Rcal}$ is not constrained to belong to $\SO$, i.e., $\widehat{\Rcal}$ does not represent a meaningful attitude estimate.
In other words, the proposed angular speed observer does not comply with the geometry of the special group $\SO$, yet, as shown in Section~\ref{sec:Stability}, it generates an asymptotically converging estimate of the angular speed $\omega$. This enables to dramatically simplify the analysis of the observer.

To analyze the considered state estimation problem, let us define the following estimation error for the angular momentum $\tilde{q}\coloneqq q-\hat{q}$. Moreover, to simplify the analysis, inspired by \cite{li2015finite,mayhew2013synergistic}, we assume the input torque $u$ is generated by the following exosystem:
\begin{equation}
\label{eq:exo}
\dot{u}\in M\mathbb{B},\qquad u\in\mathcal{K}_u,   
\end{equation}
where $M>0$ and $\mathcal{K}_u\subset\R^3$ is compact. Notice that solutions to \eqref{eq:exo} are bounded but not necessarily differentiable, which allows \eqref{eq:exo} to capture a large class of inputs.

By taking as a state $x\coloneqq (\Rcal, q, \widetilde{\Rcal},\tilde{q}, u)\in\mathcal{U}\coloneqq\R^{3\times 3}\times\R^3\times\R^{3\times 3}\times\R^{3}\times\R^{3}$, endowed with the norm 
$$\vert x\vert\coloneqq \sqrt{\langle \Rcal, \Rcal\rangle_{\mathcal{F}}+q\tr q+\langle \widetilde{\Rcal},\widetilde{\Rcal}\rangle_{\mathcal{F}}
+\tilde{q}\tr\tilde{q}+u\tr u}\,,
$$
the interconnection of plant \eqref{eq:plant} and observer \eqref{eq:obs} can be represented as the following constrained differential inclusion 
\begin{equation}
\dot{x}\in F(x),\quad x\in\mathcal{X}\,,
\label{eq:Closed-loop}
\end{equation}
where $\mathcal{X}\coloneqq \SO\times\mathcal{K}_q\times \R^{3\times 3}\times\R^3\times \mathcal{K}_u\subset\mathcal{U}$ and
\begin{equation}
F(x)\coloneqq 
\begin{bmatrix}
&[\Rcal J_0^{-1}\Rcal\tr q]_\times\Rcal\\
&u\\
&[\Rcal J_0^{-1}\Rcal\tr\tilde{q}]_\times\Rcal-\Gamma(\widetilde{\Rcal})\\
&-K\Rcal J_0^{-1}\Rcal\tr\s(\widetilde{\Rcal}\Rcal\tr-\Rcal\widetilde{\Rcal}\tr)\\
&M\mathbb{B}\end{bmatrix},\\
\quad\forall x\in\mathcal{X}.
\label{eq:Fdyn}
\end{equation}

\begin{remark}
Maximal solutions to \eqref{eq:Closed-loop} are not necessarily complete. On the other hand, any bounded solution to \eqref{eq:plant}-\eqref{eq:obs}-\eqref{eq:exo}
can be captured by \eqref{eq:Closed-loop} by taking the compact set $\mathcal{K}_q$ large enough.  
\end{remark}

The following result can be established for \eqref{eq:Closed-loop}.
\begin{lemma}
\label{lemma:WellPosed}
The set $\mathcal{X}$ is closed in $\mathcal{U}$, $F$ is locally bounded and outer semicontinuous relative to $\mathcal{X}$, and for all $x\in\mathcal{X}$, $F(x)$ is convex. Namely, \eqref{eq:Closed-loop} satisfies the basic conditions in Section~\ref{sec:prel}.
\end{lemma}
\begin{proof}
Closednees of $\mathcal{X}$ follows directly from $\SO$, $\mathcal{K}_q$, and $\mathcal{K}_u$ being compact. Local boundedness of $F$ relative to $\mathcal{X}$ follows from the continuity of the first four entries of $F$ and the compactness of the last entry, i.e., $M\mathbb{B}$. The same continuity property, along with the closedness of $M\mathbb{B}$ imply that the graph of $F$ is relatively closed in $\mathcal{X}\times \mathcal{U}$, hence from \cite[Lemma 5.10]{goebel2012hybrid} $F$ is outer semicontinuous relative to $\mathcal{X}$. To conclude, notice that for all $x\in\mathcal{X}$, $F(x)=\{f_1\}\times \{f_2\}\times \{f_3\}\times \{f_4\}\times M\mathbb{B}$ for some $f_i$'s. This shows that $F(x)$ is convex for all $x\in\mathcal{X}$. 
\end{proof}

To solve the state estimation problem, we introduce the following compact attractor
\begin{equation}
\label{eq:calA}
\A\coloneqq\{x\in\mathcal{X}\colon\widetilde{\Rcal}=0, \tilde{q}=0\}\subset\mathcal{X},
\end{equation}
and design the observer gains $\Gamma$ and $K$ to ensure that $\A$ is globally pre-asymptotically stable for \eqref{eq:Closed-loop}, thereby ensuring that $\hat{\omega}$ in \eqref{eq:obs} asymptotically approaches $\omega$ (along with complete solutions). In particular, notice that for all $x\in\mathcal{X}$, one has
$d(x,\A)=\sqrt{\vert\widetilde{\Rcal}\vert^2_\mathcal{F}+ \tilde{q}\tr\tilde{q}}$.
\section{Analysis of the Observer}
\label{sec:Stability}
\subsection{Nominal case}
Before stating the main result of this section, we recall the following obvious result that applies to \eqref{eq:plant}.
\begin{proposition}
\label{prop:Precomp}
Let $\phi$ be any complete solution to \eqref{eq:plant} and let $(\phi_{\Rcal}, \phi_q)$ be, respectively, the $\Rcal$ and $q$ components of $\phi$. Then, $(\phi_{\Rcal}, \phi_q)$ is bounded.
\end{proposition}
\begin{proof}
The proof simply follows from the fact that $\rge(\phi_{\Rcal}, \phi_q)\subset\SO\times\mathcal{K}_q$, which is a compact set.
\end{proof}
In the result given next, we provide sufficient conditions on the nonlinear gain $\Gamma$ and on the gain $K$ to ensure global pre-asymptotic stability of the set $\A$ for \eqref{eq:Closed-loop}. The result relies on the use of an invariance principle, whose application is enabled by the fact that \eqref{eq:Closed-loop} satisfies the basic conditions (see Lemma~\ref{lemma:WellPosed}); see \cite{sanfelice2007invariance} for more details on invariance principles for constrained differential inclusions.   
\begin{theorem}
\label{the:Main}
Let $\Gamma$ be positive definite, i.e., for all $R\in\R^{3\times 3}\setminus\{0\}$, $\langle \Gamma(R), R\rangle_\mathcal{F}>0$, $\Gamma(0)=0$, and $K\in\mathbb{S}_+^3$. 
Then, the set $\A$ defined in \eqref{eq:calA} is GpAS for \eqref{eq:Closed-loop}.
\end{theorem}
\begin{proof}
Consider the following Lyapunov candidate:
\begin{equation}
V(x)\coloneqq\frac{1}{2}\langle\widetilde{\Rcal}, \widetilde{\Rcal}\rangle_{\mathcal{F}}+\frac{1}{2}\tilde{q}\tr K^{-1}\tilde{q}\qquad \forall x\in\mathcal{U}.
\label{eq:V}    
\end{equation}
Notice that $V$ is continuously differentiable on $\mathcal{U}$ and positive definite with respect to the set $\A$ defined in \eqref{eq:calA} on $\mathcal{X}$.
With a slight abuse of notation, let for all $x\in\mathcal{X}$, $\dot{V}(x)\coloneqq\langle \nabla V(\widetilde{\Rcal}, \tilde{q}), h(x)\rangle$,
where
$$
h(x)\coloneqq\begin{bmatrix}
[\Rcal J_0^{-1}\Rcal\tr\tilde{q}]_\times\Rcal-\Gamma(\widetilde{\Rcal})\\
-K\Rcal J_0^{-1}\Rcal\tr\vec(\widetilde{\Rcal}\Rcal\tr-\Rcal\widetilde{\Rcal}\tr)
\end{bmatrix},
$$
and for all $(R_1, q_1), (R_2, q_2)\in\R^{3\times 3}\times\R^{3}$,
$\langle (R_1, q_1), (R_2, q_2)\rangle\coloneqq\langle R_1,R_2\rangle_{\mathcal{F}}+q_1\tr q_2.$
Then, direct computations (see also \cite[Proposition 10.7.4, page 631]{bernstein2009matrix} for trace differentiation rules) show that for all $x\in\mathcal{X}$, one has:
\begin{equation}
\label{eq:Vdot}
\begin{aligned}
\dot{V}(x)= &\Tr\left(\widetilde{\Rcal}\tr [\Rcal J_0^{-1}\Rcal\tr\tilde{q}]_\times\Rcal\right)-\Tr\left(\widetilde{\Rcal}\tr\Gamma(\widetilde{\Rcal})\right)\\
&-\tilde{q}\tr\Rcal J_0^{-1}\Rcal\tr\s(\widetilde{\Rcal}\Rcal\tr-\Rcal\widetilde{\Rcal}\tr).\end{aligned}
\end{equation}
In particular, by using the cyclic property of the trace, one has, for all $x\in\mathcal{X}$ 
\begin{equation}
\label{eq:Vdot_2}
\begin{aligned}
\dot{V}(x)=&\Tr\left([\Rcal J_0^{-1}\Rcal\tr\tilde{q}]_\times\Rcal\widetilde{\Rcal}\tr\right)-\Tr\left(\widetilde{\Rcal}\tr\Gamma(\widetilde{\Rcal})\right)\\
&-\tilde{q}\tr\Rcal J_0^{-1}\Rcal\tr\s(\widetilde{\Rcal}\Rcal\tr-\Rcal\widetilde{\Rcal}\tr).
\end{aligned}
\end{equation}
At this stage, notice that from Property~\ref{prop:TraceProp} for all $x\in\mathcal{X}$
$$
\begin{aligned}
&\Tr\left([\Rcal J_0^{-1}\Rcal\tr\tilde{q}]_\times\Rcal\widetilde{\Rcal}\tr\right)\\
&\qquad\qquad\qquad=\tilde{q}\tr\Rcal J_0^{-1}\Rcal\tr\s(\widetilde{\Rcal}\Rcal\tr-\Rcal\widetilde{\Rcal}\tr).
\end{aligned}
$$
Hence, plugging the above expression into \eqref{eq:Vdot_2} gives:
\begin{equation}
\label{eq:Vdot_3}
\dot{V}(x)=-\Tr\left(\widetilde{\Rcal}\tr\Gamma(\widetilde{\Rcal})\right)=-\langle \widetilde{\Rcal},   \Gamma(\widetilde{\Rcal})\rangle_{\mathcal{F}}\qquad\forall x\in\mathcal{X}.
\end{equation}
Using positive definiteness of $\Gamma$, one has that for all $x\in\mathcal{X}$, $\dot{V}(x)\leq 0$. Hence, since $V$ is continuous and positive definite with respect to $\A$ on $\mathcal{X}$, from \cite[Theorem 8.8]{goebel2012hybrid} it follows that $\A$ is stable for \eqref{eq:Closed-loop}. 
To conclude the proof, we show that the assumption in \cite[Theorem 8.8, item b]{goebel2012hybrid} holds. In particular, we show that for all
$r^\star>0$, the largest weakly invariant subset of
\begin{equation}
\label{eq:OmegaLimit}
V^{-1}(r^\star)\cap \dot{V}^{-1}(0)
\end{equation}
is empty. Pick $r^\star>0$ and assume by contradiction that there exists a nonempty weakly invariant set $\Omega$, such that $\Omega\subset V^{-1}(r^\star)\cap \dot{V}^{-1}(0)$. This implies that there exists a solution $\psi$ to \eqref{eq:Closed-loop} such that for all $t\in\dom\psi$, $V(\psi(t))=r^\star$ and $\dot{V}(\psi(t))=0$. Namely, \textcolor{blue}{due to $\Gamma$ being positive definite and $\Gamma(0)=0$}\footnote{\textcolor{blue}{The fact that $\Gamma(0)=0$ is not needed here but it is used below to restrict the dynamics of \eqref{eq:Closed-loop}}.}, for all $t\in\dom\psi$:
\begin{equation}
\begin{aligned}
&\frac{1}{2}\psi_{\tilde{q}}(t)\tr K^{-1}\psi_{\tilde{q}}(t)=r^\star,&\psi_{\tilde{\Rcal}}(t)=0.
\end{aligned}
\label{eq:rstarPos}    
\end{equation}
Combining the above expression with \eqref{eq:Closed-loop} implies that for all\footnote{\textcolor{blue}{$\psi_\Rcal$ and $\psi_\Rcal\tr$ are swapped in the published version.}} $t\in\dom\psi$,  $[\textcolor{blue}{\psi_\Rcal(t) J_0^{-1}\psi_\Rcal\tr(t)}\psi_{\tilde{q}}(t)]_\times\psi_\Rcal(t)=0
$. Therefore, since $\rge\psi_{\Rcal}\subset\SO$, using Property~\ref{prop:Ker}, one has that for all $t\in\dom\psi$, $\textcolor{blue}{\psi_\Rcal(t) J_0^{-1}\psi_\Rcal\tr(t)}\psi_{\tilde{q}}(t)=0$. Hence, from $J_0\succ 0$, it turns out that $\psi_{\tilde{q}}=0$. This, due to
\eqref{eq:rstarPos}, contradicts the fact that $r^\star>0$. To show global pre-attractivity, observe that from Proposition~\ref{prop:Precomp} and \eqref{eq:Vdot_3}, it follows that any solution to \eqref{eq:Closed-loop} is bounded. Therefore, from \cite[Corollary 8.4]{goebel2012hybrid} it follows that for some $r^\star\geq 0$ any complete solution to \eqref{eq:Closed-loop} converges to the nonempty largest weakly invariant subset of \eqref{eq:OmegaLimit}. Hence, since we showed above that $r^\star$ must be zero, it follows that complete solutions to \eqref{eq:Closed-loop} converge to $\A$. This establishes the result. 
\end{proof}
\subsection{Robustness to small measurement noise}
We assume that measurements of the rotation matrix $\Rcal$ are affected by a bounded additive measurement noise $\eta\in\R^{3\times 3}$. In this scenario, the interconnection of the plant and the observer can be modeled as
\begin{equation}
\dot{x}\in F_n(x, \eta)\quad x\in\mathcal{X}\cap\mathcal{O}, \eta\in\R^{3\times 3},
\label{eq:Closed-loop_noise}
\end{equation}
where for all $x\in\mathcal{X}, \eta\in\R^{3\times 3}$, \textcolor{blue}{$F_n(x, \eta)$ are the dynamics of $x$ when 
$\Rcal$ is replaced by $\Rcal+\eta$ in the observer dynamics \eqref{eq:obs}} and $\mathcal{O}\subset\mathcal{U}$ is an arbitrary compact set that is introduced to simplify the analysis. 
\begin{proposition}[Robustness to small measurement noise]
\label{prop:RobStabSmall}
Suppose that the set $\mathcal{A}$ in \eqref{eq:calA} is GpAS for system \eqref{eq:Closed-loop}. Then, there exists $\beta\in\mathcal{KL}$ such that the following holds. Let $\mathcal{O}\subset\mathcal{U}$ be compact and $\varepsilon>0$. Then, there exists $\delta>0$ such that any maximal solution pair\footnote{A pair $(\phi, \eta)$ is a solution pair to \eqref{eq:Closed-loop_noise} if it satisfies its dynamics; see, e.g., \cite{cai2009characterizations} for more details.} $(\phi, \eta)$ to \eqref{eq:Closed-loop_noise} with $\rge\eta\subset\delta\mathbb{B}$ satisfies: 
\begin{equation}
\label{eq:KL_pract}
d(\phi(t), \mathcal{A})\leq \beta(d(\phi(0),\mathcal{A}), t)+\varepsilon \quad \forall t\in\dom\phi.
\end{equation}
\end{proposition}
\begin{proof}
Let $\rho>0$ and define the following $\rho$-perturbation of \eqref{eq:Closed-loop}:
\begin{equation}
\label{eq:F_rho}
\dot{x}\in F^\prime_\rho(x),\qquad x\in\mathcal{X}\cap\mathcal{O},
\end{equation}
with $F^\prime_\rho(x)\coloneqq F(x)+\rho\mathbb{B}$. Since \eqref{eq:Closed-loop} satisfies the basic conditions and from Theorem~\ref{the:Main} $\mathcal{A}$ is GpAS for \eqref{eq:Closed-loop}, by applying \cite[Theorem 6.30, Theorem 7.12, and Lemma 7.20]{goebel2012hybrid}, it follows that there exists\footnote{A function $\beta\colon\R_{\geq 0}\times\R_{\geq 0}\rightarrow\R_{\geq 0}$ is a  class-$\mathcal{KL}$ function, also written $\beta\in\mathcal{KL}$, if it is nondecreasing in its first argument, nonincreasing in its second argument, $\lim_{r\rightarrow 0^+}\beta(r, s)=0$ for all $s\in\R_{\geq 0}$, and $\lim_{s\rightarrow \infty}\beta(r, s)=0$ for all $r\in\R_{\geq 0}$.} $\beta\in\mathcal{KL}$ and $\rho^\star>0$ such that any maximal solution $\varphi$ to \eqref{eq:F_rho} with $\rho\in(0, \rho^\star]$ satisfies $d(\varphi(t), \mathcal{A})\leq \beta(d(\varphi(0),\mathcal{A}), t)+\varepsilon$ for all $t\in\dom\varphi$, which reads as the bound in \eqref{eq:KL_pract}. To conclude, we relate solutions to \eqref{eq:Closed-loop_noise} and solutions to \eqref{eq:F_rho}. To this end, notice that by continuity of the observer dynamics and compactness of $\mathcal{X}\cap\mathcal{O}$, there exists a continuous function $\sigma\colon\R_{\geq 0}\rightarrow\R_{\geq 0}$ with $\sigma(0)=0$ such that for all $\tilde{\delta}\geq0$ and $(x, \eta)\in(\mathcal{X}\cap\mathcal{O})\times\tilde{\delta}\mathbb{B}$, $F_n(x, \eta)\subset F^\prime_{\sigma(\tilde{\delta})}(x)$. Select $\delta>0$ small enough so that $\sigma(\delta)\leq \rho^\star$. Then, given any maximal solution pair $(\phi, \eta)$ to \eqref{eq:Closed-loop_noise} with $\rge\eta\subset\delta\mathbb{B}$, $\phi$ is a solution to \eqref{eq:F_rho} with $\rho=\rho^\star$. Hence, due to the bound established above for solutions to \eqref{eq:F_rho}, \eqref{eq:KL_pract} holds for $\phi$ and this concludes the proof.  
\end{proof}
\section{Angular Speed Estimation on $\Sone$}
\label{sec:SO2}
In this section, we show how the construction proposed in this paper can be also adopted to build an angular speed observer whenever rotations occur around a fixed axis.
This problem has been recently considered in \cite{brentari2018class} via hybrid systems tools. The approach proposed by the authors in that reference essentially relies on the idea that discontinuities induced by phase wrap can be thought as instantaneous changes of the angular position. A wholly similar approach has been adopted in the literature of pulse coupled oscillators; see, e.g., \cite{ferrante2016hybrid,ferrante2017robust,nunez2016synchronization}, just to mention a few.
A similar instance of this state estimation problem has also appeared in \cite{bernard2020estimation} in the context of sensorless control of permanent magnets motors. 

To fully capitalize on the approach proposed in this paper, as opposed to \cite{bernard2020estimation,brentari2018class}, we represent the dynamics of a body spinning around a fixed axis as a dynamical system in the state space $\R^{2\times 2}\times\R$ evolving in $\Sone\times U$, where $U$ is a compact interval, namely:
\begin{equation}
\label{eq:PlantSO2}
\begin{aligned}
&\dot{R}=\omega S R,\quad \dot{\omega}=\frac{1}{J} u.
\end{aligned}    
\end{equation}
The state $R\in\R^{2\times 2}$ represents the attitude of the body that is assumed to be measured, $\omega$ is the angular speed to be estimated, $u\in\R$ is the input torque, $J>0$ the inertia of the body around the spinning axis, and\footnote{\textcolor{blue}{Typo in published version, i.e., $1$ and $-1$ swapped in $S$.}} 
$
\textcolor{blue}{S\coloneqq \begin{bmatrix}
0&-1\\
1&0\end{bmatrix}}$. For the sake of exposition, we assume that $u=0$; nonzero input torques can be easily included in the analysis, e.g., by following the approach outlined in Section~\ref{sec:ProbStatSO3}. 
\begin{remark}
\textcolor{blue}{In \cite{bernard2020estimation,brentari2018class}, the authors assume that only ``wrapped measurements'' for the angular position are available. On the other hand, 
by \textcolor{blue}{observing that there exists a natural bijection between $\Sone$ and $(-\pi, \pi]$}, 
wrapped measurements, say in $(-\pi, \pi]$, are nothing but the image of elements in $\Sone$ via the above mentioned bijection. Thus, the approach we propose in this paper can be directly employed  
in the settings analyzed in\footnote{
\textcolor{blue}{
In the published version, the statement ``$\Sone$ and $(-\pi, \pi]$ are isomorphic'' should read ``there exists a bijection between $\Sone$ and $(-\pi, \pi]$''. 
Indeed, obviously there is no continuous mapping from $\Sone$ to $(-\pi, \pi]$. Notice, our argument to relate wrapped measurements to elements in $\Sone$ does not rely on the continuity of such a mapping but only on its bijectivity; see Section~\ref{sec:disc} in the Appendix for more details. 
}
}
 \cite{bernard2020estimation,brentari2018class}.}
\end{remark}
\subsection{Angular speed observer}
Following Section~\ref{sec:ProbStatSO3}, we propose the following observer\footnote{A reduced order observer may be easily derived by replacing the measurement $R$ by $Rv$, with $0\neq v\in\R^2$. However, the resulting observer cannot be employed for angular position filtering as in Section~\ref{sec:AttiFiltering}.}:
\begin{equation}
\begin{aligned}
&\dot{\widehat{R}}=\hat{\omega}S R+\Gamma(R-\widehat{R}),\quad\dot{\hat{\omega}}=\kappa\Tr((R-\widehat{R})\tr S R)
\end{aligned}
\label{eq:ObsSO2}    
\end{equation}
where $\Gamma\in\mathcal{C}^0(\R^{2\times 2};\R^{2\times 2})$ and $\kappa\in\R$ are to be designed. By introducing the error coordinates $\widetilde{R}\coloneqq R-\widehat{R}$, $\tilde{\omega}\colon=\omega-\hat{\omega}$, the interconnection of \eqref{eq:PlantSO2} and \eqref{eq:ObsSO2} can be written as the following dynamical system with state $\chi\coloneqq (R,\omega, \widetilde{R},\tilde{\omega})\in\mathcal{V}\coloneqq\R^{2\times 2}\times U\times\R^{2\times 2}\times\R$
\begin{equation}
\left\{\begin{array}{ll}
\begin{aligned}
&\dot{R}=\omega S R\\
&\dot{\omega}=0\\
&\dot{\widetilde{R}}=\tilde{\omega}S R-\Gamma(\widetilde{R})\\
&\dot{\tilde{\omega}}=-\kappa\Tr(\widetilde{R}\tr S R)
\end{aligned} &x\in\mathcal{T}
\end{array}\right.,
\label{eq:ObsSO2--Plant}
\end{equation}
where $\mathcal{T}\coloneqq\Sone\times U\times\R^{2\times 2}\times\R$.

Similarly to Section~\ref{sec:Stability}, we introduce the  compact set:
\begin{equation}
\mathcal{W}\coloneqq\{\chi\in\mathcal{T}\colon \widetilde{R}=0, \tilde{\omega}=0 
\}\subset\mathcal{V},
\label{eq:calW}  
\end{equation}
wherein the estimation error is equal to zero and provide sufficient conditions for global pre-asymptotic stability of $\mathcal{W}$.  
\begin{theorem}
\label{thm:Theorem2}
Let $\Gamma$ be positive definite (with respect to the Frobenius inner product on $\R^{2\times 2}$), $\Gamma(0)=0$, and $\kappa>0$. Then, the set\footnote{\textcolor{blue}{Notice that maximal solutions to \eqref{eq:ObsSO2--Plant} are complete, hence in this case GpAS turns out to be GAS.}}  $\mathcal{W}$ in \eqref{eq:calW} is GpAS for \eqref{eq:ObsSO2--Plant}. 
\hfill$\diamond$
\end{theorem}
The above result can be easily proven via a simple adaptation of the proof of Theorem~\ref{the:Main}, therefore its proof is omitted.
\subsection{Discussion on Angular Position Filtering}
\label{sec:AttiFiltering}
The objective of the proposed observer is to generate an estimate of the angular speed based on attitude measurements. In this section, we briefly discuss about the problem of attitude filtering, which is a relevant problem in the presence of noisy measurements. 
As mentioned in Section~\ref{sec:ProbStatSO3}, the proposed observer does not provide a meaningful estimate of the attitude. 
A possible approach to solve this problem consists of computing online the ``nearest'' rotation matrix to the estimate generated by the observer. However, obtaining a closed-form expression of such a ``projection'' operation turns out to be a difficult problem; see \cite{sarabandi2020closed}. Nonetheless, in this section we show how such an explicit expression can be obtained in the context of $\Sone$.

The set of nearest matrices to a given matrix $H\in\R^{2\times 2}$ can be implicitly represented by the set-valued map $\Pi\colon\R^{2\times 2}\rightrightarrows\Sone$ defined for all $H\in\R^{2\times 2}$ as:
\begin{equation}
\Pi(H)\!\!\coloneqq\!\underset{R}{\argmin}\!\left\{\vert R-H\vert^2_{\mathcal{F}}\,\,\vert\,\,  R\tr R\!=\!\Id, \det(R)=1\right\}.
\label{eq:Pi}
\end{equation}
The result given next provides a closed-form expression of the set valued map $\Pi$.  
\begin{proposition}
Let $\Upsilon\coloneqq\{H\in\R^{2\times 2}\, \vert\, H=H\tr, \Tr(H)=0\}$. Define, for all $H\in\R^{2\times 2}\setminus\Upsilon$:
$$
\begin{aligned}
&\vartheta(H)\coloneqq\frac{h_{11}+h_{22}}{\sqrt{{\left(h_{11}+h_{22}\right)}^2+{\left(h_{12}-h_{21}\right)}^2}},\\
&\digamma(H)\coloneqq\frac{h_{12}-h_{21}}{\sqrt{{\left(h_{11}+h_{22}\right)}^2+{\left(h_{12}-h_{21}\right)}^2}}
\end{aligned}
$$
where for all $(i,j)\in\{1,2\}^2$, $h_{ij}$ is the $i,j$-entry of $H$. Then, the following identity holds:
\begin{equation}
\Pi(H)=\begin{cases}
\Sone&\text{if}\,H\in\Upsilon\\
\left\{\begin{bmatrix}
\vartheta(H)&\digamma(H)\\
-\digamma(H)&\vartheta(H)
\end{bmatrix}\right\}&\text{otherwise}.
\end{cases}
\label{eq:PiExp}
\end{equation}
\end{proposition}
\begin{proof}
The proof of the result is established applying Lagrange multiplier method. In particular, by recalling that $\Sone=\left\{\begin{bmatrix}
x&y\\
-y&x
\end{bmatrix}\colon x^2+y^2=1, x,y\in\R
\right\}$, the optimization problem associated to \eqref{eq:Pi} can be equivalently rewritten as:
\begin{equation}
\label{eq:Opti}
\begin{aligned}
&\underset{x,y\in\R}{\minimize}\frac{(h_{21} + y)^2 + (h_{11} - x)^2 + (h_{22} - x)^2 + (h_{12}- y)^2}{2}\\
&\quad\text{subject\,\,to}\,\,x^2+y^2=1
\end{aligned}
\end{equation}
The Lagrangian associated to \eqref{eq:Opti} reads:
$$
\begin{aligned}
\mathcal{L}_H(R,\!\lambda)\!=\!&\underbrace{\frac{(h_{21}\!+\! y)^2 + (h_{11}\! -\! x)^2 + (h_{22}\!-\! x)^2 + (h_{12}\!-\! y)^2}{2}}_{g(R)}\\
&+\lambda (1-x^2-y^2)
\end{aligned}
$$
where $\lambda\in\R$ is the Lagrange multiplier. Therefore, the necessary conditions for optimality read:
 \begin{equation}
 \begin{aligned}
 &\textcolor{black}{(2-\lambda^\star)x^\star-h_{11}- h_{22}=0,}
 &(2-\lambda^\star)y^\star + h_{21}-h_{12}=0\\
 &{x^\star}^2+{y^\star}^2=1
 \end{aligned}
 \label{eq:LagrangeConds}
\end{equation}
Now we analyze two possible cases:
\begin{itemize}
\item Case I: $H\in\Upsilon$. In this case, solving \eqref{eq:LagrangeConds} one gets
$(2-\lambda^\star)x^\star=0$, $(2-\lambda^\star)y^\star=0$,
${x^\star}^2+{y^\star}^2=1$,
which gives $R^\star\in\Sone, \lambda^\star=2$, thereby showing \eqref{eq:PiExp} for all $H\in\Upsilon$.
\smallskip

\item Case II: $H\notin\Upsilon$. In this case \eqref{eq:LagrangeConds} yields $x^\star_{1,2}=\pm\vartheta(H),\quad y^\star_{1,2}=\pm\digamma(H)$, $\lambda^\star_{1,2}=2\pm\sqrt{(h_{11}+ h_{22})^2+(h_{12}- h_{21})^2}$.
At this stage, notice that
$$
\begin{aligned}
&g\left(\left[\begin{smallmatrix}
-\vartheta(H)&-\digamma(H)\\
\digamma(H)&-\vartheta(H)
\end{smallmatrix}\right]\right)-g\left(\left[\begin{smallmatrix}
\vartheta(H)&\digamma(H)\\
-\digamma(H)&\vartheta(H)
\end{smallmatrix}\right]
\right)\\
&\quad\qquad=-\frac{4\left({\left(h_{11}+h_{22}\right)}^2+{\left(h_{12}-h_{21}\right)}^2\right)}{\sqrt{{\left(h_{11}+h_{22}\right)}^4+{\left(h_{12}-h_{21}\right)}^4}}\leq 0
\end{aligned}
$$
Namely, the unique element of \eqref{eq:Pi} reads as $R^\star=\left[\begin{smallmatrix}
\vartheta(H)&\digamma(H)\\
-\digamma(H)&\vartheta(H)
\end{smallmatrix}\right]$. This concludes the proof.
\end{itemize}
\end{proof}
The property below holds for set valued map $\Pi$.
\begin{lemma}
\label{lemm:DiffPi}
Let $R\mapsto\pi(R)\in\Pi(R)$ be any selection of $\Pi$ and $\mathcal{Q}\coloneqq\R^{2\times 2}\setminus\Upsilon$. Then, $\pi$ is continuously differentiable on $\mathcal{Q}$. \QEDB
\end{lemma}
\begin{proof}
The proof simply follows by noticing that $\mathcal{Q}$ is open, for all $H\in\mathcal{Q}$, $
\pi(H)=\underbrace{\left[\begin{smallmatrix}
\vartheta(H)&\digamma(H)\\
 -\digamma(H)&\vartheta(H)
 \end{smallmatrix}\right]}_{T(H)}
$, and that $H\mapsto T(H)$ is continuously differentiable.
\end{proof}
Lemma~\ref{lemm:DiffPi} ensures that for all $\R_{\geq 0}\ni t\mapsto R(t)\in\Sone$ and $\R_{\geq 0}\ni t\mapsto \widetilde{R}(t)$, the following implication holds:
$$
\lim_{t\rightarrow\infty}\widetilde{R}(t)=0\implies
\lim_{t\rightarrow\infty}\Pi(R(t)-\widetilde{R}(t))-R(t)=0
$$
Roughly speaking, this ensures that projecting $\widehat{R}$ onto $\Sone$   yields a proper converging estimate of the attitude $R$.
\section{Numerical Examples}
\label{sec:Ex}
\subsection{Angular Speed estimation on $\SO$}
\label{sec:ExSO3}
In this example, we showcase the effectiveness of the results proposed in Section~\ref{sec:Stability} via numerical simulations. In particular, we consider the scenario of \cite[Section V]{wu2015angular} in which $J_0=\begin{bmatrix}
5 &0& 0\\ 0& 1& 0\\ 0& 0& 2
\end{bmatrix}$ and \eqref{eq:plant} is initialized as $\Rcal(0)=\exp(\pi/4 [e_1]_\times)$, $\omega(0)=(1,-1.5, 2.5)$, where $e_1=(1, 0, 0)$. The  observer \eqref{eq:obs} is initialized at $\widehat{\Rcal}(0)=\Rcal_0, \hat{q}(0)=0$ and for simplicity we set $u=0$. \figurename~\ref{fig:omega3} depicts the evolution of  $\omega-\hat{\omega}$ with $\Gamma=20\id$ and different selections of $K$ and $\Gamma$. \figurename~\ref{fig:omega3} clearly shows that the proposed observer yields a converging estimate of $\omega$ and underlines the impact of the parameters $K$ and $\Gamma$ on the convergence speed. Namely, for $K=100 J_0$ and $\Gamma=20\id$ the proposed observer practically converges in $1.5$ seconds, which appears to be much faster when compared with the numerical results in \cite{wu2015angular}. Notice also that for the same selection of $K$, taking $\Gamma=10^3\id$ negatively affects the convergence speed.  
\begin{figure}[h]
\centering
\psfrag{t}[1][1][1]{$t$ [sec]}
\psfrag{d1}[l][l][0.82]{$K_1$}
\psfrag{d2}[l][l][0.82]{$K_2$}
\psfrag{d3}[l][l][0.82]{$K_3$}
\psfrag{d4}[l][l][0.82]{$K_4$}
\psfrag{y1}[1][1][1]{$\omega_1-\hat{\omega}_1$}
\psfrag{y2}[1][1][1]{$\omega_2-\hat{\omega}_2$}
\psfrag{y3}[1][1][1]{$\omega_3-\hat{\omega}_3$}
\includegraphics[trim=0.1cm 0.1cm 0.1cm 0.1cm, clip, width=\columnwidth]{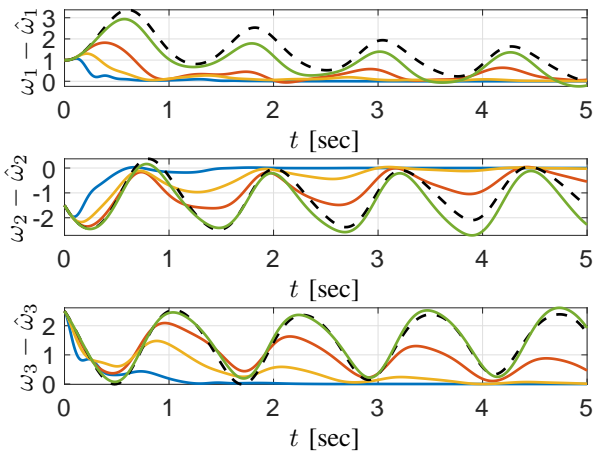}
\caption{Evolution of the angular speed estimation error for the example in Section~\ref{sec:ExSO3} for different tunings: $K_1=100 J_0, \Gamma=20\id$ (blue), $K_2=10 J_0, \Gamma=20\id$ (orange), $K_3=30 J_0, \Gamma=20\id$ (yellow), $K_4=5 \Id, \Gamma=20\id$ (dashed black), and  $K_4=100 J_0, \Gamma=10^3\id$ (green).}
\label{fig:omega3}
\end{figure}
To assess the impact of measurement noise on the observer, in \figurename~\ref{fig:noiseSO3} we report the evolution of $\omega-\hat{\omega}$ (after transient) in the presence of a band-limited white measurement noise. The figure underlines that fast convergence speed comes at the price of increased noise sensitivity.  
\begin{figure}[h]
\psfrag{t}[1][1][1]{$t$ [sec]}
\psfrag{y1}[1][1][1]{$\omega_1-\hat{\omega}_1$}
\psfrag{y2}[1][1][1]{$\omega_2-\hat{\omega}_2$}
\psfrag{y3}[1][1][1]{$\omega_3-\hat{\omega}_3$}
\psfrag{d1}[l][l][0.82]{$K_1$}
\psfrag{d2}[l][l][0.82]{$K_2$}
\psfrag{d3}[l][l][0.82]{$K_3$}
\includegraphics[trim=0.1cm 0.1cm 0.1cm 0.1cm, clip, width=\columnwidth]{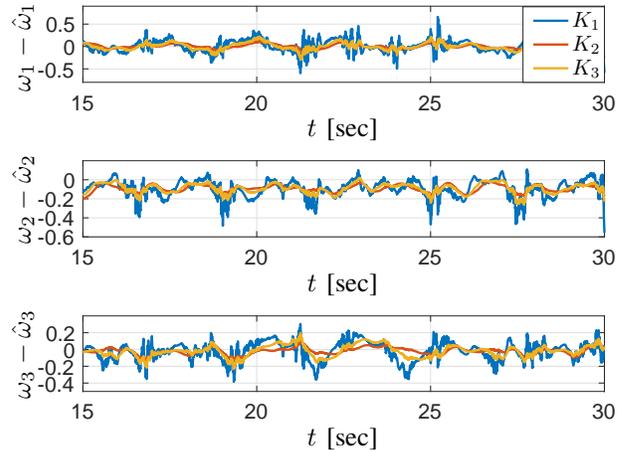}
\caption{Evolution of angular speed estimation error for the example of Section~\ref{sec:ExSO3} with $\Gamma=20\id$ for different selections of $K$: $K_1=100 J_0$, $K_2=10 J_0$, and $K_3=30 J_0$ in the presence of measurement noise. Simulations are performed with a band-limited white noise with noise power $10^{-5}$.}
\label{fig:noiseSO3}
\end{figure}

\subsection{Angular speed estimation on $\Sone$}
\label{ex:SO2}
In this section, we illustrate the results of Section~\ref{sec:SO2} via numerical simulations. As in \cite{bernard2020estimation,brentari2018class}, we assume to have access to measurements of the angular position $\theta$ wrapped in the interval $(-\pi, \pi]$ and map those to $\Sone$ via the standard \textcolor{blue}{bijection} from $(-\pi, \pi]$ to $\Sone$. \figurename~\ref{fig:omegaSO2} shows the evolution of $\theta, \omega$, $\hat{\omega}$, and $\hat{\theta}$ from the following initial condition: $\theta(0)=\pi/2$, $\widehat{R}(0)=\Id$, $\omega(0)=10$, $\hat{\omega}(0)=0$, when $\Gamma=40\id$ and $K=200$. The estimate $\hat{\theta}$ is obtained via \eqref{eq:PiExp} by mapping $\Pi(\widehat{R})\in\Sone$ to $(-\pi, \pi]$. The picture points out that, despite wrapped measurements, the proposed observer provides a convergent estimate of the angular speed.
To illustrate the results of Section~\ref{sec:AttiFiltering}, we perform some simulations in which angular measurements are affected by a bounded additive noise $\eta(t)=0.1\sin(10^4 t)$. \figurename~\ref{fig:omegaSO2Noise} reports the evolution of the estimate $\hat{\theta}$ of the angular position along with the angular speed estimation error. Simulations show that the proposed observer is robust to measurement noise and that the attitude filtering technique of Section~\ref{sec:AttiFiltering} is effective in yielding a smoother estimate of the angular position.
\begin{figure}[ht]
\psfrag{t}[1][1][1]{$t$ [sec]}
\psfrag{d1}[1][1][1]{$\omega$}
\psfrag{d2}[1][1][1]{$\hat{\omega}$}
\psfrag{p1}[1][1][1]{$\theta$}
\psfrag{p2}[1][1][1]{$\hat{\theta}$}
\psfrag{v}[1][1][1]{$\omega,\hat{\omega}$}
\psfrag{pi}[1][1][1]{$\pi$}
\psfrag{-pi}[1][1][1]{$-\pi$}
\psfrag{pi/2}[1][1][1]{$\frac{\pi}{2}$}
\psfrag{-pi/2}[1][1][1]{$-\frac{\pi}{2}$}
\psfrag{p}[1][1][1]{$\theta, \hat{\theta}$}
\includegraphics[trim=0cm 0cm 0cm 0cm, clip,width=\columnwidth]{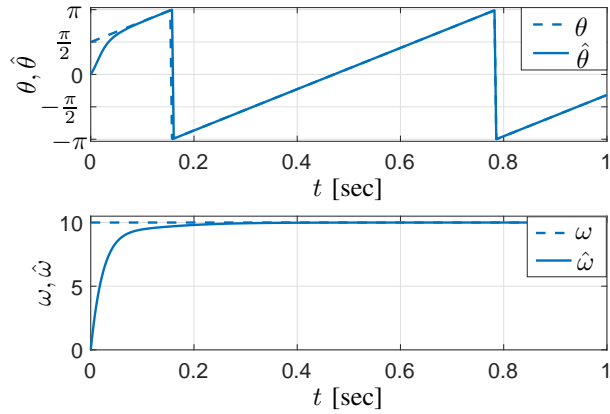}
\caption{Evolution of $\omega$, $\hat{\omega}$, and of the measured angular position $\theta$, and its estimate $\hat{\theta}$ for the  example in Section~\ref{ex:SO2}.}
\label{fig:omegaSO2}
\end{figure}
\begin{figure}[h!]
\vspace{-0.05cm}
\psfrag{t}[1][1][1]{$t$ [sec]}
\psfrag{v}[1][1][1]{$\omega-\hat{\omega}$}
\psfrag{pi}[1][1][1]{$\pi$}
\psfrag{-pi}[1][1][1]{$-\pi$}
\psfrag{pi/2}[1][1][1]{$\frac{\pi}{2}$}
\psfrag{-pi/2}[1][1][1]{$-\frac{\pi}{2}$}
\psfrag{d1}[1][1][1][0]{$\theta$}
\psfrag{d2}[1][1][1][0]{$\hat{\theta}$}
\psfrag{p}[1][1][1]{$\theta, \hat{\theta}$}
\includegraphics[trim=0cm 0cm 0cm 0cm, clip, width=\columnwidth]{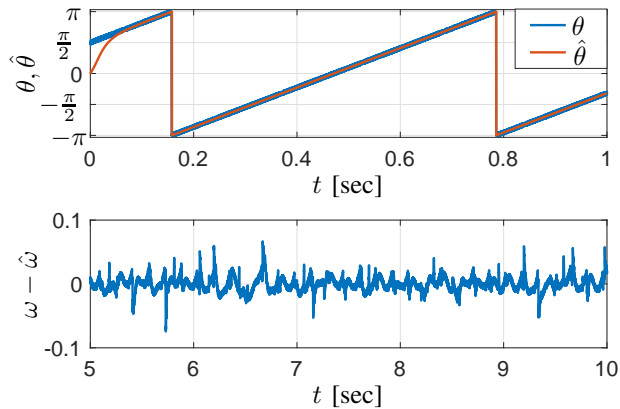}
\caption{Simulations in the presence of noisy measurements for the example in Section~\ref{ex:SO2}.}
\label{fig:omegaSO2Noise}
\end{figure}
\section{Conclusion}
In this letter, we provided new insights on the angular speed estimation problem for rigid bodies. The main contribution is an angular speed observer freed from the typical geometrical constraints induced by the manifold $\SO$. Global asymptotic stability of the estimation error dynamics was established via a Barbasin-Krasovskii-LaSalle principle. The methodology was specialized to the case of systems rotating around a fixed axis. In this specific case, we showed how the proposed observer can also be used to filter noisy angular measurements. Future research directions include an extension towards {\em exponential} convergence as in \cite{SarrasECC20} and a finer analysis of the effect of measurement noise. In addition, experimental validation and sampled-data implementations will be considered. 
\bibliographystyle{plain}
\balance
\bibliography{biblio}
\appendix
\subsection{Additional details for the proof of Theorem~\ref{the:Main}}
The proof of Theorem~\ref{the:Main} relies on the computation of vector $\nabla V(\widetilde{\Rcal}, \tilde{q})$, 
where 
$$
(\widetilde{\Rcal}, \tilde{q})\mapsto V(\widetilde{\Rcal}, \tilde{q})=\frac{1}{2}\langle \widetilde{\Rcal},\widetilde{\Rcal}\rangle_{\mathcal{F}}+\frac{1}{2} \tilde{q}\tr K^{-1}\tilde{q} 
$$
which enables to express the directional derivative of $x\mapsto V(\widetilde{\Rcal}, \tilde{q})$ along the dynamics of system \eqref{eq:Closed-loop} as
$$
\langle \nabla V(\widetilde{\Rcal}, \tilde{q}), h(x)\rangle
$$
where for all $x\in\mathcal{X}$
$$
h(x)\coloneqq\begin{bmatrix}
[\Rcal J_0^{-1}\Rcal\tr\tilde{q}]_\times\Rcal-\Gamma(\widetilde{\Rcal})\\
-K\Rcal J_0^{-1}\Rcal\tr\vec(\widetilde{\Rcal}\Rcal\tr-\Rcal\widetilde{\Rcal}\tr)
\end{bmatrix},
$$
and for all $(R_1, q_1), (R_2, q_2)\in\R^{3\times 3}\times\R^{3}$,
$\langle (R_1, q_1), (R_2, q_2)\rangle\coloneqq\langle R_1,R_2\rangle_{\mathcal{F}}+q_1\tr q_2.$

By definition
$$
\nabla V(\widetilde{\Rcal}, \tilde{q})=\sum_{i=1}^{12}(DV(\widetilde{\Rcal}, \tilde{q})e_i)e_i\in\R^{3\times 3}\times \R^{3}
$$
where $\{e_i\}_{i=1}^{12}$ is any orthonormal basis of $\R^{3\times 3}\times \
\R^3$ (endowed with the standard inner product $\langle \cdot, \cdot\rangle_{\mathcal{F}}+\langle \cdot, \cdot\rangle_{\R^3}$). In particular, $\nabla V(\widetilde{\Rcal}, \tilde{q})$ is the Riesz representation of $DV(\widetilde{\Rcal}, \tilde{q})$, which, being unique,  does not depend on the basis $\{e_i\}_{i=1}^{12}$. To perform this computation, we select the canonical basis of $\R^{3\times 3}\times\
\R^3$, that is:
$$
\left(\{E_i\}_{i=1}^{9}\times\{0\}\right)\cup \left(\{0\}\times\{g_i\}_{i=1}^{3}\right)
$$ 
where $\{E_i\}$ and $\{g_i\}$ are, respectively, the canonical basis of $\R^{3\times 3}$ and $\R^{3}$. Following this approach, it easy to see that for 
all $(e, 0)\in \left(\{E_i\}_{i=1}^{9}\times\{0\}\right)$, one has
$$
(DV(\widetilde{\Rcal}, \tilde{q})(e, 0))(e, 0)=\langle \widetilde{\Rcal}, e\rangle_{\mathcal{F}}
$$
and that for all $(0, w)\in \left(\{0\}\times\{g_i\}_{i=1}^3\right)$
$$
(DV(\widetilde{\Rcal}, \tilde{q})(0,w))(0, w)=\tilde{q}\tr K^{-1} w
$$
Thus
$$
\nabla V(\widetilde{\Rcal}, \tilde{q})=(\widetilde{\Rcal}, K^{-1}\tilde{q})\in\R^{3\times 3}\times\R^3
$$
\subsection{Discussion on wrapped angular measurements.}
The set $\Sone$ is isomorphic to the set of equivalence classes $[-\pi, \pi]/\sim$, where $\sim$ is the ``gluing'' equivalence: $x\sim x^\prime$ if $x=x^\prime$ or $x=-\pi$ and $x^\prime=\pi$, or $x=\pi$ and $x^\prime=-\pi$. The isomorphism is given by
$$
\begin{aligned}
&([-\pi, \pi]/\sim)\ni [x]\mapsto f(x)=\begin{bmatrix}
\cos(x)&-\sin(x)\\
\sin(x)&\cos(x)
\end{bmatrix}\in\Sone\\
&\Sone\ni y\mapsto f^{-1}(y)=
\begin{cases}
\{\mathrm{atan2}(y_{21}, y_{11})\}&\text{if}\,\, y_{21}\neq 0\\
\{\pm\pi\}&\text{else}
\end{cases}
\end{aligned}
$$ 
where $[x]$ stands for the equivalence class of $x\in[-\pi, \pi]$ and $\mathrm{atan2}$ stands for the 2-argument arctangent function. Clearly $f$ is continuous, while continuity of $f^{-1}$ can be proven by observing that $f^{-1}$ sends closed sets in $\Sone$ into closed sets into $[-\pi, \pi]/\sim$.  
In this sense, wrapped measurements of the angular position can be thought as the image in $(-\pi, \pi]$ of points in $([-\pi, \pi]/\sim)$ through the map:
$$
g([x])=\begin{cases}
x&\text{if}\,\,[x]\cap\{\pm\pi\}=\emptyset\\
\pi&\text{else}
\end{cases}
$$     
It can be shown that $g$ is discontinuous and this is what leads to ``jumps'' in wrapped measurements.
\label{sec:disc}
 \end{document}